\documentclass[11pt,a4paper,reqno]{amsart}

\usepackage{amsmath,amssymb,amsfonts,amsthm,mathtools,bm}
\usepackage{setspace}
\usepackage{xcolor,graphicx}
\usepackage{mathrsfs}%
\usepackage{pifont}
\usepackage[top=1.5in, bottom=1.2in,
  left=1in, right=1in]{geometry}
\usepackage[ruled,noline]{algorithm2e}
\usepackage[mathscr]{euscript}
\usepackage[round,authoryear]{natbib}
\usepackage{hyperref}
\hypersetup{
    colorlinks = true,   
    citecolor = blue,    
    linkcolor = red,     
    urlcolor = green     
}
\usepackage{tabularx}
\usepackage{makecell}
\usepackage{bbm}
\usepackage{subcaption}
\usepackage{multirow}%
\usepackage{amsthm}%
\usepackage{mathrsfs}%
\usepackage[title]{appendix}%
\usepackage{xcolor}%
\usepackage{textcomp}%
\usepackage{setspace}
\usepackage{manyfoot}%
\usepackage{booktabs}%
\usepackage{algpseudocode}%
\usepackage{booktabs}
\usepackage{pifont}
\usepackage{subcaption}%
\usepackage{makecell}
\usepackage{tcolorbox}
\usepackage[ruled,noline]{algorithm2e}
\newcommand{\cmark}{\ding{51}} 
\newcommand{\xmark}{\ding{55}} 

\newcommand{\SSGEN}{{\tt{SSGEN}}}

\newcommand{\mv}[1]{\boldsymbol{#1}}

\newcommand{\ww}{\boldsymbol{w}}

\newcommand{\xx}{\boldsymbol{x}}
\newcommand{\XX}{\boldsymbol{\xi}}

\newcommand{\yy}{\boldsymbol{y}}
\newcommand{\zz}{\boldsymbol{z}}

\newcommand{\VaR}{\textsc{VaR}}
\newcommand{\CVaR}{\textsc{CVaR}}

\newcommand{\Pphi}{\boldsymbol{\Phi}}
\newcommand{\pphi}{\boldsymbol{\phi}}

\newcommand{\fgen}{f_{\textsf {gen}}}

\newtheorem{assumption}{Assumption}
\newtheorem{theorem}{Theorem}[section]
\newtheorem{lemma}{Lemma}[section]
\newtheorem{proposition}{Proposition}[section]
\newtheorem{corollary}{Corollary}[section]
\renewcommand{\arraystretch}{1.5}

\theoremstyle{thmstyleone}%

\theoremstyle{thmstyletwo}%

\theoremstyle{thmstylethree}%

\title[Learning Stress Laws]{An Extreme Value Perspective on Learning Stress Laws}
\author{{\large M\MakeLowercase{antu}
    G\MakeLowercase{upta}}}
    \author{{\large A\MakeLowercase{nand}
    D\MakeLowercase{eo}}}
  \address{Indian Institute of Management, Bannerghatta Road, Bangalore 560076} \email{  mantu.gupta23@iimb.ac.in, anand.deo@iimb.ac.in}
\begin{document}
\maketitle
\begin{abstract}
We introduce Self-Similar Generative Estimation (SS-GEN), a method for simulating multivariate tail events and estimating rare-event probabilities in both heavy and light-tailed settings. SS-GEN exploits asymptotic tail structure to decompose the tail distribution into an explicit radial component and a nonparametric angular component, reducing tail learning to a compact-domain problem that can be handled by off-the-shelf deep generative models. The resulting sampler generates representative extreme scenarios and supports probability estimation far beyond the observed data. Under mild nonparametric tail assumptions, we show that the SS-GEN density is asymptotically exact in the tail, with vanishing uniform relative error for regularly varying distributions and vanishing uniform log-relative error for Weibull-type distributions. Unlike existing approaches that rely on specialized architectures or parametric tail specifications, SS-GEN leverages asymptotic tail structure to enable standard generative models to generate representative extreme samples and estimate rare-event probabilities beyond the observed data.
\end{abstract}
\section{INTRODUCTION}
\label{sec:intro}
Understanding extreme events is central to risk management in finance, operations research, climate science, power systems, and safety-critical engineering (see \cite{Glasserman1}; \cite{blanchet2024efficient};  \cite{buishand2008spatial};  \cite{bienstock2014chance}; \cite{choi2024reduced}). Accurate modeling of such events requires understanding the behavior of the data-generating distribution in its tail. Yet extreme events are rarely observed, so available datasets typically contain only limited tail information. This scarcity makes tail modeling statistically difficult and can lead to substantial misestimation of rare-event probabilities.

These challenges motivate data-driven methods that can learn tail behavior from limited observations and extrapolate beyond the observed sample. In recent years, deep generative models (DGMs) have emerged as flexible tools for learning complex, high-dimensional distributions directly from data, including Variational Autoencoders (VAEs) \cite{Kingma2013}, Generative Adversarial Networks (GANs) \cite{goodfellow2014generative}, Normalizing Flows (NFs) \cite{Rezende2015}, and Diffusion Models \cite{Sohl-Dickstein2015,Ho2020}. Despite their  empirical success, however, standard DGMs often struggle to capture tail behavior accurately (see \cite{wiese2019copula,liu2024learning}). A key reason is that they are typically trained to fit the bulk of the distribution, so rare tail events have limited influence on the learned model. This suggests that the main obstacle is not a lack of expressive power, but the absence of an appropriate structural reduction of the tail-generation problem.

This paper addresses that challenge by exploiting asymptotic structure in the tail distribution. The central question we ask is: \emph{can tail structure be used to reduce extreme-sample generation to a simpler learning problem that standard DGMs can solve?} We answer the above question via the following contributions.

\noindent \textbf{i) A general representation for tail densities: }We show that, under a regularity condition, the tail density admits an asymptotic spliced radial-angular representation. The radial component captures extremal magnitude, while the compactly supported angular component captures tail dependence and can be learned non-parametrically.  Under regular variation (Assumption~\ref{assume:ht_data}), the radial component is asymptotically Pareto and independent of the angular law. Under Weibull-type tails (Assumption~\ref{assume:lt_data}), it is conditionally Gamma, with parameters determined by the angular law. We use "Self-Similar" to denote this stabilization: it is driven by Pareto scaling in heavy-tailed regimes, and by direction-dependent logarithmic decay in Weibull-type regimes.

\noindent \textbf{ii) A tail sampler based on standard DGMs: }The representation in \textbf{(i)} makes standard DGMs natural for tail generation: they learn only the angular law, while the radial law is explicit. We analyze an idealized SS-GEN sampler with known threshold angular density and prove asymptotic exactness of the spliced density, in uniform relative error for regularly varying tails and uniform log-relative error for Weibull-type tails. In practice, the angular law is learned non-parametrically using an off-the-shelf DGM, and the radial component is modeled by a low-dimensional parametric specification. Because angular learning is done at intermediate rarity levels with enough exceedances for training, SS-GEN extrapolates to regions of probability order $n^{-\kappa}$, for suitable $\kappa>1$. 

\noindent \textbf{Literature review and positioning:} Prior work on using DGMs for tail modeling, especially in heavy-tailed settings, falls into three broad strands. The first modifies standard generative architectures to explicitly encode heavy-tailed behavior. In the GAN setting, \cite{allouche2022ev} introduce generators with smooth bounded tail-index functions. \cite{cont2022tail} develops a framework for simulating multi-asset market scenarios using a bespoke loss function targeting tail risk through VaR and CVaR, although its objective differs from sampling from a multivariate extreme distribution.

The second strand modifies the latent prior itself. Examples include heavy-tailed latent distributions for GANs \cite{huster2021pareto,li2024generalized,girard:hal-04700084}, heavy-tailed priors and conditional structures in VAEs \cite{lafon2023vae,kim2024t3variationalautoencoderlearningheavytailed}, and heavy-tailed base noise in diffusion models \cite{liu2024learning,pandey2024heavytaileddiffusionmodels}.

The third strand exploits structural representations of tail behavior to simplify generation. Related examples include several recent works \cite{bhatia2021exgan,boulaguiem2022modeling,lhaut2025wasserstein}, which use tail-specific structure or decompositions to facilitate extreme-sample generation. Our approach is closest in spirit to this line of work, but differs in two important respects. First, rather than modifying architectures or replacing standard latent priors, our decomposition reduces the problem to learning a compactly supported angular component. Thus, our contribution is not a new generative architecture, but a structural reduction that makes standard DGMs applicable to tail generation. Second, unlike related approaches in this line of work, we provide approximation guarantees for the resulting density. Table~\ref{tab:rw_min_compare} summarizes representative papers from each strand and compares them only along the dimensions most relevant to our contribution; it is not intended as an exhaustive survey.

\begin{table}[htpb]
\centering
\small
\renewcommand{\arraystretch}{0.9}
\caption{Comparison with related work. ``Approximation Guarantees'' indicates whether the paper establishes asymptotic approximation of the learned tail model, or its induced density, to the target tail law.}
\begin{tabular}{lccc}
\toprule
\textbf{Paper}
& \begin{tabular}[c]{@{}c@{}}\textbf{Specialized Architecture }\\\textbf{Required?}\end{tabular}
& \begin{tabular}[c]{@{}c@{}}\textbf{Tail-Modified}\\\textbf{Prior Required?}\end{tabular}
& \begin{tabular}[c]{@{}c@{}}\textbf{Approximation}\\\textbf{Guarantees?}\end{tabular} \\
\midrule

\textbf{This paper}
& \xmark
& \xmark
& \cmark\\
\textbf{\cite{allouche2022ev}}  & \cmark
& \xmark
& \xmark\\
\textbf{\cite{huster2021pareto}}  & \cmark
& \cmark
& \xmark\\
\textbf{\cite{boulaguiem2022modeling}}  & \xmark
& \xmark
& \xmark\\
\textbf{\cite{lhaut2025wasserstein}}  & \xmark
& \cmark
& \xmark\\
\bottomrule
\end{tabular}
\renewcommand{\arraystretch}{1}
\label{tab:rw_min_compare}
\end{table}
\noindent\textbf{Outline:} The rest of the paper is structured as follows: Section~\ref{sec:assump_model}  formulates the problem and states the assumptions under which we prove our results. Sections~\ref{sec:ht_sim} and \ref{sec:lt_sim} develop our theory for the heavy-tailed (Pareto) and light-tailed (Weibullian) settings respectively. We demonstrate the applicability of our theory via experiments in Section~\ref{sec:numericals}. Proofs of theoretical results are given in Section~\ref{sec:proofs}.

\section{Problem Description and Assumptions}\label{sec:assump_model}
Let $\XX$ be a random vector in $\mathbb R^d$ with density $f_{\XX}$ supported on a cone $\mathcal E \subset \mathbb R^d$. Write $(R,\Pphi)=(\|\XX\|,\XX/\|\XX\|)$ for its radial-angular decomposition, and let $f_{R,\Pphi}$ denote the corresponding joint density. Our goal is to construct a family of generative approximations $\{\fgen^{(t)}:t>0\}$ for the tail density, indexed by a radial threshold parameter $t$. The role of $t$ is asymptotic: as $t\to\infty$, $\fgen^{(t)}$ should approximate $f_{R,\Pphi}$ uniformly in the tail under an error criterion appropriate to the tail regime. Thus, the mathematical object of interest is the tail density itself, rather than the law on any single rare set. To cover the heavy and light-tailed settings of interest, we impose the following assumptions on $f_{\XX}$.

\begin{assumption}[\textbf{Regularly-varying data}]\label{assume:ht_data}
The density of $\XX$ satisfies, uniformly over compact subsets of $\mathcal E \setminus \{\mv 0\}$,
\begin{equation}\label{eqn:ht_data}
n^{s+d} f_{\XX}(n\zz) \to \varphi^\star(\zz)
\qquad \text{as } n\to\infty,
\end{equation}
for some $s>0$ and some function $\varphi^\star : \mathcal E \to (0,\infty)$.
\end{assumption}

\begin{assumption}[\textbf{Weibull-type data}]\label{assume:lt_data}
The density of $\XX$ satisfies, uniformly over compact subsets of $\mathcal E \setminus \{\mv 0\}$,
\begin{equation}\label{eqn:lt_data}
\frac{-\log f_{\XX}(n\zz)}{n^{\gamma}} \to \Lambda^\star(\zz)
\qquad \text{as } n\to\infty,
\end{equation}
for some $\gamma>0$ and some function $\Lambda^\star : \mathcal E \to (0,\infty)$.
\end{assumption}

Assumption~\ref{assume:ht_data} covers a broad class of heavy-tailed densities with regularly varying tails, including multivariate $t$, Pareto, and Cauchy-type models, as well as mixtures built from such distributions. Assumption~\ref{assume:lt_data} captures an analogous class of lighter-tailed models with Weibull-type decay, including Gaussian and Weibull laws, as well as broad subclasses of elliptical distributions. The assumptions are nonparametric tail-domain stability conditions: no joint parametric family is imposed, and dependence is captured by the limiting angular profile. In practice, radial tail estimates, angular summaries, and conditional angular samples should be stable across high thresholds. SS-GEN may fail when tail mechanisms switch with level, mixture components dominate at different extreme scales, or directional decay oscillates; it is intended for settings with a stable radial-angular tail representation.

\begin{lemma}\label{lem:radius_angle_split}
The joint density of $(R,\Pphi)$ satisfies
\begin{equation}\label{eqn:radius_angle_split}
    f_{R,\Pphi}(r,\pphi) = r^{d-1} f_{\XX}(r\pphi).
\end{equation}
\end{lemma}

The identity in \eqref{eqn:radius_angle_split} is the starting point for the construction of SSGEN. Our development proceeds in two steps. First, we derive a tail asymptotic expansion of \eqref{eqn:radius_angle_split} using the assumed regularity of $f_{\XX}$. Second, we use this expansion to construct a tractable generative approximation to the tail density.
\section{Simulation of Heavy Tailed Distributions}\label{sec:ht_sim}
We now specialize the representation in Lemma~\ref{lem:radius_angle_split} to the heavy-tailed regime of Assumption~\ref{assume:ht_data}. 
\begin{proposition}\label{prop:Polar-to-MRV}
Let Assumption~\ref{assume:ht_data} hold. Then the conditional joint density of $(R,\Pphi)$ given $R>t$ satisfies
\begin{equation}\label{eqn:R-Phi_density}
    f_{R,\Pphi}(r,\pphi \mid R>t)
    =
    s t^{s} r^{-(s+1)} f^{(t)}(\pphi)\bigl(1+\varepsilon_t(r,\pphi)\bigr),
    \qquad r>t,
\end{equation}
where $f^{(t)}$ denotes the conditional density of $\Pphi$ given $R>t$, and $\sup_{r>t,\ \pphi\in \mathcal E \cap S_\infty^{d-1}} |\varepsilon_t(r,\pphi)| \to 0$ as  $t\to\infty.$
\end{proposition}
Proposition~\ref{prop:Polar-to-MRV} shows that, beyond a large radial threshold $t$, the conditional density of $(R,\Pphi)$ is asymptotically separable into a Pareto radial component and the conditional angular density $f^{(t)}(\pphi)$. This motivates the spliced approximation 
\begin{equation}\label{eqn:SSGEN_Ht}
    \fgen^{(t)}(r,\pphi)
    =
    f_{R,\Pphi}(r,\pphi)\,\mathbf{1}(r<t)
    +
    p_t\, s t^s r^{-(s+1)} f^{(t)}(\pphi)\,\mathbf{1}(r>t),
\end{equation}
where $p_t := \mathbb P(R>t)$.
If the threshold $t$ is chosen so that the conditional angular density $f^{(t)}(\pphi)$ can be learned from the available data, then \eqref{eqn:SSGEN_Ht} provides a candidate generative approximation to the tail density. The next theorem shows that this approximation is asymptotically exact in uniform relative error, thereby reducing the approximation of the full tail density to the problem of learning the angular component alone.

\begin{theorem}\label{thm:heavy_vanishing_RE}
Assume that Assumption~\ref{assume:ht_data} holds, and let $\fgen^{(t)}$ be defined by \eqref{eqn:SSGEN_Ht}. Then, as $t\to\infty$,
\begin{equation}\label{eqn:heavy_vanishing_RE}
    \sup_{(r,\pphi)}
    \left|
    \frac{f_{R,\Pphi}(r,\pphi)}{\fgen^{(t)}(r,\pphi)} - 1
    \right|
    \to 0.
\end{equation}
\end{theorem}

Theorem~\ref{thm:heavy_vanishing_RE} shows that approximating the tail density reduces to learning the angular component $f^{(t)}$ on a compact domain, while the radial behavior is supplied analytically by the tail asymptotics. Thus, tail learning is converted from an unbounded density-estimation problem on $\mathcal E$ into a bounded-domain nonparametric problem. Algorithm~\ref{algo:generate_extremes} implements this decomposition in practice.

\begin{algorithm}[h] 
  \caption{$\SSGEN$-HT: Sampling from the conditional tail model for regularly varying data}  \label{algo:generate_extremes}
  \KwIn{Data samples $\XX_1, \ldots \XX_n$, threshold $t$, DGM mechanism
  }
\noindent\textbf{I) Training Step}
\begin{enumerate}
    \item[(i)]  Filter out the tail samples as $\mathcal Y = \{i: R_i \geq t\}$. Construct the dataset $ \Theta_{\mathcal Y} = \left\{\Pphi_i = \frac{\XX_i}{\|\XX_i\|} : i \in \mathcal Y \right\}.$
    \item[(ii)] Train a deep generative model on the dataset $\Theta_{\mathcal Y}$.
\end{enumerate}
\noindent\textbf{II) Generation Step:} 
\begin{enumerate}
    \item[(i)] Sample $\Pphi^*$ from the trained DGM.
    \item[(ii)] Sample $R_t$ from the Pareto exceedance density $s t^s r^{-(s+1)}\mathbf 1(r>t)$.
\end{enumerate}
\noindent\textbf{Output: }Sample $H_t = R_t\cdot\Pphi^*$ 
\end{algorithm}
\noindent \textbf{Generating samples from unobserved tail regions: } Algorithm~\ref{algo:generate_extremes} requires learning the conditional angular density only at an intermediate threshold $t$, namely under the law of $\Pphi$ given $R>t$. More generally, the same mechanism applies to a broader class of rare sets lying sufficiently deep in the tail. For concreteness, we state the result for the canonical loss exceedance set
\begin{equation}\label{eqn:Ru}
\mathcal R_u := \{\zz : L(\zz) > u\},
\end{equation}
where $L:\mathcal E\to\mathbb R$ is a loss function; examples include network losses, ReLU-network outputs, linear portfolio losses, and objectives defined through linear programs. The next corollary extends Theorem~\ref{thm:heavy_vanishing_RE} to densities conditioned on $\mathcal R_u$, where in practice there may be little to no observed data.
\begin{corollary}\label{cor:gen_extreme_samples}
Suppose that, for all sufficiently large $u$ and $t$, $\mathcal R_u \subset \{\zz:\|\zz\|>t\}$. Under the assumptions of Theorem~\ref{thm:heavy_vanishing_RE}, the approximation in \eqref{eqn:heavy_vanishing_RE} continues to hold when $f_{R,\Pphi}$ and $\fgen^{(t)}$ are replaced by their conditional densities given $\mathcal R_u$.
\end{corollary}
Corollary~\ref{cor:gen_extreme_samples} shows that the SSGEN approximation remains accurate after conditioning on a sufficiently deep rare set; Experiment 2 illustrates this by comparing the distribution of a second portfolio loss conditional on a stress event in a first portfolio against a benchmark conditional tail law.
\section{Simulation of Weibullian Distributions} \label{sec:lt_sim}
We now specialize the representation in Lemma~\ref{lem:radius_angle_split} to the Weibull-type regime of Assumption~\ref{assume:lt_data}. The next result shows that, deep in the tail, the conditional density of $(R,\Pphi)$ given $R>t$ again admits an asymptotically separable form, with a nonparametric angular component and an explicit radial component whose shape depends on the angular direction through $\Lambda^\star(\pphi)$. In particular, after a suitable transformation, the radial term is of Gamma type, which again leads to a tractable generative construction.

\begin{proposition}\label{prop:Polar-to-Gamma}
If the density of $\XX$ satisfies Assumption~\ref{assume:lt_data}, then the conditional joint density of $(R,\Pphi)$ given $R>t$ admits the representation
\begin{equation}\label{eqn:Polar-to-Gamma}
    f_{R,\Pphi}(r,\pphi \mid R>t)
    =
    f^{(t)}(\pphi)\,
    \frac{
        r^{d-1}\exp\!\bigl(-r^{\gamma}\Lambda^\star(\pphi)(1+\delta(r,\pphi))\bigr)
    }{
        \int_t^\infty y^{d-1}\exp\!\bigl(-y^{\gamma}\Lambda^\star(\pphi)(1+\delta(y,\pphi))\bigr)\,dy
    }\,
    \mathbf 1(r>t),
\end{equation}
where $\sup_{r>t,\ \pphi\in \mathcal E\cap \mathcal S^{d-1}}
|\delta(r,\pphi)| \to 0 \text{ as } t\to\infty$,
and $f^{(t)}$ is the conditional density of $\Pphi$ given $R>t$.
\end{proposition}

Proposition~\ref{prop:Polar-to-Gamma} shows that, beyond a large threshold $t$, the conditional density of $(R,\Pphi)$ is asymptotically separable into the angular density $f^{(t)}(\pphi)$ and an explicit direction-dependent radial component governed by $\Lambda^\star(\pphi)$. This yields the spliced approximation 
\begin{equation}\label{eqn:SSGEN_lt}
    \fgen^{(t)}(r,\pphi)
    =
    f_{R,\Pphi}(r,\pphi)\,\mathbf{1}(r<t)
    +
    p_t\, f^{(t)}(\pphi)\,
    \frac{
        r^{d-1}\exp\!\bigl(-r^{\gamma}\Lambda^\star(\pphi)\bigr)
    }{
        \int_t^\infty y^{d-1}\exp\!\bigl(-y^{\gamma}\Lambda^\star(\pphi)\bigr)\,dy
    }
    \,\mathbf{1}(r>t),
\end{equation}
where $p_t:=\mathbb P(R>t)$.
If the threshold $t$ is chosen so that the conditional angular density $f^{(t)}(\pphi)$ can be estimated from the available exceedances, then \eqref{eqn:SSGEN_lt} provides a candidate generative approximation to the tail density. The next theorem shows that this approximation is asymptotically exact in uniform log-relative error, thereby reducing approximation of the full tail density to the problem of learning the angular component alone.

\begin{theorem}\label{thm:light_vanishing_LE}
Assume that Assumption~\ref{assume:lt_data} holds, and let $\fgen^{(t)}$ be defined by \eqref{eqn:SSGEN_lt}. Then, as $t\to\infty$,
\begin{equation}\label{eqn:light_vanishing_LE}
    \sup_{(r,\pphi)}
    \left|
    \frac{\log f_{R,\Pphi}(r,\pphi)}{\log \fgen^{(t)}(r,\pphi)} - 1
    \right|
    \to 0.
\end{equation}
\end{theorem}

Theorem~\ref{thm:light_vanishing_LE} shows that the statistical task again reduces to estimating the angular density $f^{(t)}$. As in the heavy-tailed case, this is a compact-domain learning problem, while the radial behavior is supplied analytically through the Weibull-type tail asymptotics. Standard DGMs can therefore be used to learn the angular law, with $\Lambda^\star$ governing the direction-dependent radial decay. Algorithm~\ref{algo:generate_extremes_lt} implements this decomposition in practice.

\begin{algorithm}[h]
  \caption{$\SSGEN$-LT: Sampling from the conditional tail model for Weibull-type data}
  \label{algo:generate_extremes_lt}
  \KwIn{Data samples $\XX_1,\ldots,\XX_n$, threshold $t$, DGM mechanism, estimate of $\Lambda^\star$}
  
\noindent\textbf{I) Training Step:}
Train a DGM on the angular exceedance sample $\Theta_{\mathcal Y}$ defined as in Algorithm~\ref{algo:generate_extremes}.

\noindent\textbf{II) Generation Step:}
\begin{enumerate}
\item Sample $\Pphi^\ast$ from the trained DGM.
\item Sample $Z\sim \Gamma(d/\gamma,\Lambda^\star(\Pphi^\ast))$ and set $R=Z^{1/\gamma}$.
\item Repeat until $R>t$, and return $G_t=R\Pphi^\ast$.
\end{enumerate}
\end{algorithm}

\begin{lemma}\label{lem:lt_correct}
Assume that Step~(i) of Algorithm~\ref{algo:generate_extremes_lt} produces a sample $\Pphi^*$ from the angular density $f^{(t)}$, and that $\Lambda^\star$ is known exactly. Then the output $G_t$ of Algorithm~\ref{algo:generate_extremes_lt} has density
\[
(r,\pphi)\mapsto
f^{(t)}(\pphi)\,
\frac{
r^{d-1}\exp\!\bigl(-r^\gamma \Lambda^\star(\pphi)\bigr)
}{
\int_t^\infty y^{d-1}\exp\!\bigl(-y^\gamma \Lambda^\star(\pphi)\bigr)\,dy
}
\,\mathbf 1(r>t),
\]
which is precisely the conditional density induced by the tail component of $\fgen^{(t)}$ on $\{r>t\}$.
\end{lemma} 

Lemma~\ref{lem:lt_correct} verifies that Algorithm~\ref{algo:generate_extremes_lt} samples exactly from the tail component of $\fgen^{(t)}$ when $\Lambda^\star$ is known. One possible fully data-driven light-tailed implementation estimates
\(\Lambda^\star\) through local angular smoothing. For directions near \(\phi\), one may fit a
truncated Weibull/Gamma-type radial tail rate using the exceedances \((R_i,\Phi_i)\) with \(R_i>t\) and \(\Phi_i\) in a local angular neighborhood. The fitted rates can then be smoothed across nearby directions and used as a
plug-in estimate of \(\Lambda^\star(\phi)\) in Algorithm~2. We defer tuning and consistency analysis to the full version.

\noindent\textbf{Application to Estimating Tail Probabilities:} As an application, we consider rare-event probability estimation. The next result shows that the probability assigned to $\mathcal R_u$ by the SSGEN model is asymptotically accurate as the event becomes rarer.
\begin{theorem}\label{thm:est_rare_probs}
Let $P_{\mathrm{gen}}(\mathcal R_u)$ denote the probability of $\mathcal R_u$ under the SSGEN generative approximation. Consider a regime in which $u\to\infty$ and $t=t(u)\to\infty$, with $t(u) = o(u)$. Then:
\begin{enumerate}
    \item[(i)] If Assumption~\ref{assume:ht_data} holds and SSGEN is defined through Algorithm~\ref{algo:generate_extremes}, then
    \[
        \left|
        \frac{P_{\mathrm{gen}}(\mathcal R_u)}{P(\mathcal R_u)} - 1
        \right|
        \to 0.
    \]

    \item[(ii)] Suppose Assumption~\ref{assume:lt_data} holds and SSGEN is defined through Algorithm~\ref{algo:generate_extremes_lt}. Further, suppose that the loss $L(\cdot)$ satisfies $n^{-1}L(n\zz_n) \to L^\star(\zz)$ for some positive $L^\star$ whenever $\zz_n\to\zz$. Then,
    \[
        \left|
        \frac{\log P_{\mathrm{gen}}(\mathcal R_u)}{\log P(\mathcal R_u)} - 1
        \right|
        \to 0.
    \]
\end{enumerate}
\end{theorem}
To interpret this result, suppose $n$ observations are available and that
$p_u:=\mathbb P(\mathcal R_u)=n^{-\kappa}$ for some $\kappa>1$. Then the expected number of observations in $\mathcal R_u$ is $n p_u=n^{1-\kappa}\to 0$, so direct learning in this region is infeasible. If, however, the intermediate threshold $t$ is chosen so that $p_t:=\mathbb P(R>t)=n^{-q}$ for some $q\in(0,1)$, then the number of exceedances above $t$ is of order $n^{1-q}\to\infty$, and the angular density $f^{(t)}$ remains learnable from the data (see, e.g., \cite{liang2021well} for learning rates). SS-GEN exploits this separation of scales by learning the angular law at threshold $t$ and combining it with the explicit radial tail model to extrapolate to the much rarer region $\mathcal R_u$. As a result, it supports accurate estimation of tail probabilities and related risk measures in regions that are effectively unobserved in the original dataset.
\section{Numerical Experiments}\label{sec:numericals}
We now present numerical illustrations of the proposed methodology. In this version of the paper, we restrict attention to the heavy-tailed setting owing to space constraints; a broader empirical study, including light-tailed examples, is deferred to a full version of the paper.
As a heavy-tailed test bed, we consider a $d=10$ dimensional mixture model
\[
\XX = B\XX_1 + (1-B)\XX_2,
\]
where $B\sim \mathrm{Bernoulli}(1/2)$, $\XX_1=\mv U \Lambda R$ with $\mv U$ uniformly distributed on the unit sphere and $R$ Pareto with tail index $s=2.5$, and $\XX_2$ is multivariate Student-$t$ with $\nu=3$ degrees of freedom. This construction combines heterogeneous heavy-tailed behavior with nontrivial dependence, providing a somewhat representative setting to test the performance of SSGEN. For the DGM in Algorithm~\ref{algo:generate_extremes}, we use a standard GAN trained on the angular exceedance sample $\Theta_{\mathcal Y}$, with latent vectors drawn uniformly from the unit sphere. The angular generator is a fully connected network with dense layers, Leaky ReLU activations, batch normalization, and a \texttt{tanh} output layer. Since \texttt{tanh} bounds the output but does not enforce unit norm, each generated angular vector is normalized before radial recombination. The projected Wasserstein diagnostic is computed on these normalized directions, matching the angular samples used by SS-GEN.

\noindent\textbf{Experiment 1: Tail Risk Estimation.}
As a first illustration of the utility of SS-GEN, we consider estimation of the value at risk (VaR) and conditional value at risk (CVaR) of the portfolio loss $L(\XX)=\ww^\top \XX$, where
\begin{equation}\label{eqn:var_cvar}
    \VaR_{\alpha}(L)
    =
    \inf\{u: P(L>u)\le 1-\alpha\},
    \qquad
    \CVaR_{\alpha}(L)
    =
    E\!\left[L \mid L\ge \VaR_{\alpha}(L)\right].
\end{equation}

Since both quantities are determined by the tail of $L(\XX)$, we begin by estimating the exceedance probability $P(L>u)$ by combining the empirical contribution below the threshold $t$ with Monte Carlo samples from the SS-GEN tail model above $t$. Specifically, we set
\[
    \hat F_n(u)
    =
    \frac{1}{n}\sum_{i=1}^n \mv{1}\!\left(L(\XX_i)>u,\ \|\XX_i\|\le t\right)
    +
    \hat p_t^{(n)}\frac{1}{M}\sum_{i=1}^M \mv{1}\!\left(L(R_{t,i}\Pphi_i^\star)>u\right),
\]
where $\{(R_{t,i},\Pphi_i^\star)\}_{i=1}^M$ are generated using Algorithm~\ref{algo:generate_extremes}: the directions $\Pphi_i^\star$ are sampled from the learned angular law above threshold $t$, and the radii $R_{t,i}$ are sampled from the Pareto exceedance distribution on $(t,\infty)$. Here $\hat p_t^{(n)}$ denotes the empirical estimate of $p_t=\mathbb P(R>t)$. The tail index $s$ is estimated from the exceedance set $\mathcal Y$ using the Hill estimator, since the tail cdf of $R$ is regularly varying under Assumption~\ref{assume:ht_data}:
\[
    \hat s
    =
    \left(
    \frac{1}{|\mathcal Y|}
    \sum_{j\in\mathcal Y}
    \log\!\left(\frac{R_j}{t}\right)
    \right)^{-1}.
\]
The VaR and CVaR estimates are then obtained from the estimated tail probabilities $\hat F_n$.

We take $n=1000$ observations and choose $t$ as the empirical $90$th percentile of $\{\|\XX_i\|\}_{i=1}^n$, yielding $|\mathcal Y|=100$ exceedances. We generate $M=10000$ tail samples to evaluate $\hat F_n(u)$ and repeat the resulting Monte Carlo estimation step $1000$ times for the downstream VaR/CVaR calculation. Since training the DGM in Step~I of Algorithm~\ref{algo:generate_extremes} takes approximately $20$ minutes, the generator is trained once on an independent sample of size $n=1000$. Accordingly, the reported variability reflects Monte Carlo uncertainty conditional on the trained generator, rather than end-to-end retraining variability. Summary statistics are reported in Table~\ref{tab:var_cvar_results}, while Figure~\ref{fig:var_cvar_boxplot} shows the distribution of estimates across replications. To examine sensitivity to $|\mathcal Y|$, we separately take $n=1000$ observations and vary the percentile threshold to obtain $|\mathcal Y|\in\{50,100,150,200\}$, repeating the CVaR estimation at $\alpha=0.999$ over $1000$ replications for each setting (right panel of Figure~\ref{fig:var_cvar_boxplot}).

\begin{table}[htbp]
\centering
\caption{Summary of VaR and CVaR estimates over $1{,}000$ replications. ``Med.'' denotes the median estimate, ``95\% CI'' the empirical $95\%$ interval across replications, and ``\% Err.'' the relative error of the median estimate with respect to the benchmark value.}
\label{tab:var_cvar_results}
\small
\begin{tabular}{c r r r r @{\hspace{1em}} r r r r}
\toprule
& \multicolumn{4}{c}{VaR} & \multicolumn{4}{c}{CVaR} \\
\cmidrule(r){2-5} \cmidrule(l){6-9}
$\alpha$ & True & Med. & 95\% CI & \% Err. & True & Med. & 95\% CI & \% Err. \\
\midrule
0.95 & 2.48 & 2.55 & (2.33, 2.82) & 3.0 & 4.33 & 4.47 & (3.81, 5.24) & 3.2 \\
0.98 & 3.85 & 3.82 & (3.44, 4.23) & 0.9 & 6.26 & 6.48 & (5.23, 8.12) & 3.5 \\
0.99 & 5.11 & 5.08 & (4.35, 5.83) & 0.5 & 8.13 & 8.59 & (6.51, 11.29) & 5.7 \\
0.995 & 6.70 & 6.74 & (5.49, 8.15) & 0.7 & 10.48 & 11.38 & (8.07, 16.01) & 8.6 \\
0.999 & 11.92 & 13.03 & (9.11, 18.09) & 9.3 & 18.66 & 21.85 & (12.81, 37.27) & 17.1 \\
\bottomrule
\end{tabular}
\end{table}

\begin{figure}[htbp]
    \centering
    \includegraphics[width=0.9\linewidth]{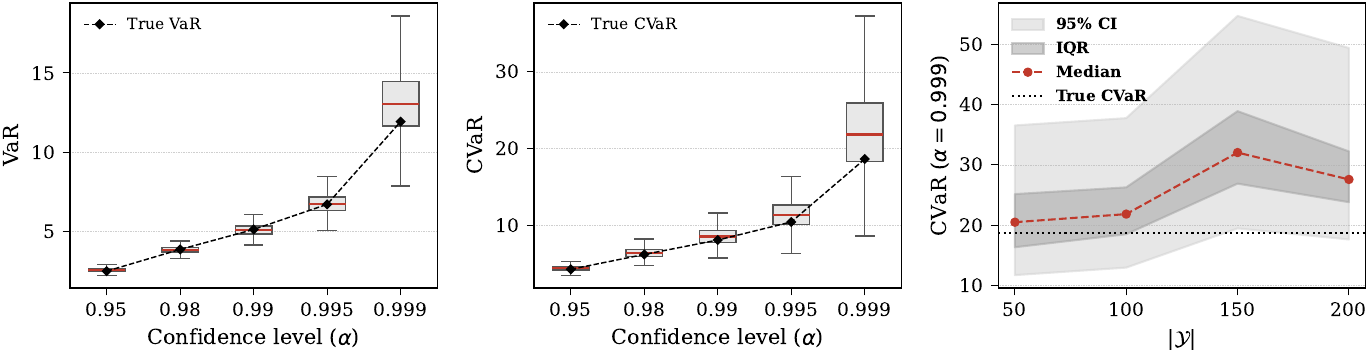}
    \caption{Left: distribution of VaR estimates across confidence levels. Middle: distribution of CVaR estimates across confidence levels. Right: CVaR estimates at $\alpha = 0.999$ as a function of $|\mathcal{Y}|$ over $1000$ replications.}
    \label{fig:var_cvar_boxplot}
\end{figure}
Our results show that SS-GEN estimates VaR and CVaR accurately across the range of confidence levels considered. Performance is strongest at moderate tail levels, with noticeable but still controlled degradation only at the most extreme level $\alpha=0.999$. Further, the performance of our procedure is robust over a range of intermediate thresholds considered. Preliminary $d=20$ experiments show the expected threshold tradeoff. At $\alpha=0.95$, VaR (resp. CVaR) relative errors were $1.15\%$ (resp. $6.71\%$) with the 80th-percentile threshold, versus $11.78\%$ (resp. $23.26\%$) with the 90th-percentile threshold. The less aggressive threshold retains more angular exceedances and gives more stable training; the point is threshold sensitivity, not dimension-free scalability.

\noindent\textbf{Experiment 2: Conditional Tail Distribution Quality.} We next evaluate how well SS-GEN captures the conditional tail law
\[
P\!\left(\XX\in\cdot \mid L_i(\XX)\ge \VaR_{\alpha}(L_i)\right).
\]
Unlike Experiment~1, which focuses on scalar tail-risk functionals, this experiment tests whether SS-GEN preserves the effect of stress in one portfolio on the distribution of an overlapping portfolio. To this end, we condition on exceedance of one portfolio loss and compare the induced distribution of a second portfolio loss under the generated and benchmark tail samples. Let
\[
L_i(\XX)=\ww_i^\top \XX,\qquad L_j(\XX)=\ww_j^\top \XX,
\]
where $\ww_i$ and $\ww_j$ have partially overlapping supports. For a fixed confidence level $\alpha$, let $\VaR_{\alpha}(L_i)$ be a high-accuracy benchmark estimate computed from a large independent Monte Carlo sample from the true mixture model. We define the benchmark conditional tail sample as
\[
\mathcal T_{\mathrm{ref}}
=
\left\{
\XX_k: L_i(\XX_k)\ge \VaR_{\alpha}(L_i),\; k=1,\dots,N_{\mathrm{ref}}
\right\},
\]
where $\{\XX_k\}_{k=1}^{N_{\mathrm{ref}}}$ is a benchmark sample of size $10^6$. 
We compare three methods for generating samples satisfying $L_i(\tilde\XX)\ge \VaR_{\alpha}(L_i)$:
\begin{enumerate}
    \item \textbf{SS-GEN.} Generate $M=10{,}000$ tail samples using Algorithm~\ref{algo:generate_extremes} and retain those satisfying the stress condition.
    \item \textbf{Fitted multivariate $t$.} Fit a multivariate $t$ distribution to the $n$ observations, draw samples from the fitted model, and retain those satisfying the stress condition.
    \item \textbf{Na\"ive resampling.} Draw an independent sample of size $n$ from the true mixture model and retain only observations satisfying the stress condition. This baseline reflects direct empirical tail sampling without extrapolation.
\end{enumerate}


These parametric and empirical baselines are chosen to isolate the effect of radial-angular extrapolation. We omit tail-specific DGM baselines because a fair comparison would require architecture-specific tuning or modified priors; the present comparison instead contrasts SS-GEN with the two most direct alternatives: parametric heavy-tail fitting and empirical tail resampling. For each method, the retained tail samples are projected onto the second portfolio axis via
$
L_j(\tilde \XX)=\ww_j^\top \tilde \XX$, and compared against the benchmark projection
$L_j(\XX_k),$ for $ \XX_k\in \mathcal T_{\mathrm{ref}}$. 
We quantify discrepancy using the two-sample Kolmogorov--Smirnov distance
$D_{n,M}= \sup_x
\left| \hat F_{\mathrm{ref}}(x)-\hat F_{\mathrm{gen}}(x)
\right|$,
where $\hat F_{\mathrm{ref}}$ and $\hat F_{\mathrm{gen}}$ denote the empirical distribution functions of the projected benchmark and generated tail samples, respectively.

We consider two experimental setups: Setup A with  $n=500$ observations, and  Setup B with  $n=1000$ observations. In each case $t$ is chosen as the $90$th percentile of the samples $\|\XX_i\|$. 
In both setups, we consider confidence levels $\alpha\in\{0.95,0.99\}$ and repeat the entire procedure $1000$ times.

\begin{table}[ht]
\centering
\caption{K-S distance between the projected conditional tail law and the benchmark reference over $1000$ repetitions. Smaller values indicate closer agreement. Entries report mean (standard deviation).}
\label{tab:ks_distance}
\small
\begin{tabular}{lccc}
\toprule
 & \textbf{SS-GEN} & \textbf{Fitted multivariate $t$} & \textbf{Na\"ive resampling} \\
\midrule
\multicolumn{4}{l}{\textit{Setup A} ($n=500$, $|\mathcal Y|=50$)} \\[2pt]
$\alpha=0.95$ & $\mathbf{0.071\ (0.002)}$ & $0.078\ (0.004)$ & $0.171\ (0.053)$ \\
$\alpha=0.99$ & $\mathbf{0.115\ (0.003)}$ & $0.138\ (0.005)$ & $0.380\ (0.140)$ \\[6pt]
\multicolumn{4}{l}{\textit{Setup B} ($n=1000$, $|\mathcal Y|=100$)} \\[2pt]
$\alpha=0.95$ & $\mathbf{0.041\ (0.003)}$ & $0.125\ (0.005)$ & $0.119\ (0.039)$ \\
$\alpha=0.99$ & $\mathbf{0.058\ (0.003)}$ & $0.151\ (0.005)$ & $0.270\ (0.102)$ \\
\bottomrule
\end{tabular}
\end{table}

Table~\ref{tab:ks_distance} shows that SS-GEN consistently produces the closest approximation to the benchmark conditional tail law across both sample sizes and both confidence levels. The advantage is particularly pronounced in the more data-scarce regime $\alpha=0.99$, where naive resampling deteriorates substantially and the fitted multivariate $t$ model exhibits a persistent bias.

\section{Conclusion}
This paper introduces a framework for simulating multivariate tail events and estimating rare-event probabilities in both heavy-tailed and light-tailed settings by decomposing the tail distribution into an analytically known radial component and a nonparametric angular component learned from exceedance data via deep generative models. 
Unlike prior generative approaches to extremes, our method comes with formal approximation guarantees. Numerical experiments on CVaR estimation corroborate our theoretical findings.

\bibliography{ref}
\bibliographystyle{plainnat}
\newpage
\appendix
\section{Proofs}\label{sec:proofs}
\noindent \textbf{Proof of Lemma~\ref{lem:radius_angle_split}:}  Consider the transformation
$T(\zz)=\bigl(\|\zz\|,\zz/\|\zz\|\bigr)=(r,\pphi)$.
Using the standard Jacobian change-of-variables formula,
$f_{R,\Pphi}(r,\pphi)
=
f_{\XX}(\zz)/|J(\zz)|$.
Now, $\zz = r\pphi$ by definition, and therefore, $|J(\zz)| = \|\zz\|^{1-d}$
Therefore $f_{R,\Pphi}(r,\pphi)
=
f_{\XX}(r\pphi)\,r^{d-1}$.\qed

\noindent \textbf{Proof of Proposition~\ref{prop:Polar-to-MRV}:} By Lemma~\ref{lem:radius_angle_split},
$f_{R,\Pphi}(r,\pphi)=r^{d-1}f_{\XX}(r\pphi)$.
Hence, by Assumption~\ref{assume:ht_data}, 
$f_{R,\Pphi}(r,\pphi)
=
r^{-(s+1)}\varphi^\star(\pphi)\bigl(1+\varepsilon(r,\pphi)\bigr)$,
where
$\eta_t:=\sup_{r>t,\ \pphi\in S^{d-1}_\infty} |\varepsilon(r,\pphi)| \to 0$ as  $t\to\infty$.
Now, 
\begin{align*}
    f_{R,\Pphi}(r,\pphi\mid R>t) &= f^{(t)}(\pphi) \frac{f_{R,\Pphi}(r,\pphi)}{\int_t^\infty f_{R,\Pphi}(y,\pphi) dy} \mv{1}(r>t) \\ &=f^{(t)}(\pphi)  \frac{r^{-(s+1)} \varphi^\star(\pphi) (1+\varepsilon(r,\pphi)) }{\int_t^\infty y^{-(s+1)} \varphi^\star(\pphi) (1+\varepsilon(y,\pphi)) dy }\mv{1}(r>t) \\
    &= st^s f^{(t)}(\pphi) r^{-(s+1)} \frac{(1+\varepsilon(r,\pphi))}{st^s \int_t^\infty y^{-(s+1)}  (1+\varepsilon(y,\pphi)) dy}
\end{align*}
Since $|\varepsilon(y,\pphi)|\le \eta_t$ for all $y>t$, uniformly in $\pphi$, $
1-\eta_t
\le
st^s\int_t^\infty y^{-(s+1)}(1+\varepsilon(y,\pphi))\,dy
\le
1+\eta_t$,
because $st^s\int_t^\infty y^{-(s+1)}\,dy=1$. Thus the denominator is $1+o(1)$ uniformly in $\pphi$, while the numerator is also $1+o(1)$ uniformly over $r>t,\pphi$. Defining
\[
\varepsilon_t(r,\pphi):=
\frac{1+\varepsilon(r,\pphi)}
{st^s\int_t^\infty y^{-(s+1)}(1+\varepsilon(y,\pphi))\,dy}-1
\]
we obtain
$
f_{R,\Pphi}(r,\pphi\mid R>t)
=
st^s f^{(t)}(\pphi) r^{-(s+1)}\bigl(1+\varepsilon_t(r,\pphi)\bigr)$.   Moreover, since both numerator and denominator lie in \([1-\eta_t,1+\eta_t]\), $\sup_{r>t,\pphi}\bigl|\varepsilon_t(r,\pphi)\bigr|
\le
\tfrac{2\eta_t}{1-\eta_t}\to 0,$
which proves the result.\qed

\noindent \textbf{Proof of Theorem~\ref{thm:heavy_vanishing_RE}:} From \eqref{eqn:SSGEN_Ht},
\[
\frac{f_{R,\Pphi}(r,\pphi)}{\fgen^{(t)}(r,\pphi)} - 1
=
\begin{cases}
0, & \text{if } r\le t,\\[1ex]
\begin{aligned}
\frac{f_{R,\Pphi}(r,\pphi)}
{p_t\, s t^s r^{-(s+1)} f^{(t)}(\pphi)} - 1
\end{aligned}
& \text{if } r>t.
\end{cases}
\]
Then, 
\[
\sup_{(r,\pphi)} \left|\frac{f_{R,\Pphi}(r,\pphi)}{\fgen^{(t)}(r,\pphi)} - 1\right| =  \sup_{r>t,\pphi}  \left| \frac{p_t f_{R,\Pphi}(r,\pphi \mid R>t)}{p_t st^s r^{-(s+1)} f^{(t)}(\pphi)}- 1\right|  = \sup_{r>t,\pphi}  \left| \frac{f_{R,\Pphi}(r,\pphi \mid R>t)}{st^s r^{-(s+1)} f^{(t)}(\pphi)}- 1\right| 
\]
The last term converges to $0$ as $t\to\infty$ owing to Proposition~\ref{prop:Polar-to-MRV}, concluding the proof.\qed

\noindent \textbf{Proof of Proposition~\ref{prop:Polar-to-Gamma}:} Recall that under Assumption~\ref{assume:lt_data}, for $\pphi\in \mathcal S^{d-1}\cap \mathcal E$ (a compact set), 
$\log f(r\pphi)  = -r^{\gamma} \Lambda^\star(\pphi) (1+\delta(r,\pphi))$, where $\sup_{r>t,\pphi} |\delta(r,\pphi)| \to 0 $ as $t\to\infty$. Then, 
\begin{align*}
    f_{R,\Pphi}(r,\pphi\mid R>t)  = f^{(t)}(\pphi) \frac{r^{d-1} \exp(-r^{\gamma} \Lambda^\star(\pphi) (1+\delta(r,\pphi)))}{\int_t^\infty y^{d-1} \exp(-y^{\gamma} \Lambda^\star(\pphi) (1+\delta(y,\pphi)) dy} \mv{1}(r>t). \qed
\end{align*}

\noindent\textbf{Proof of Theorem~\ref{thm:light_vanishing_LE}:} By definition of $f^{(t)}_{\mathrm{gen}}$, it
suffices to consider $r>t$. Define
\[
I^{\delta}_{t}(\phi):=\int_{t}^{\infty}y^{d-1}
\exp\bigl(-y^{\gamma}\Lambda^{\star}(\phi)(1+\delta(y,\phi))\bigr)\,dy,
\qquad
I_{t}(\phi):=\int_{t}^{\infty}y^{d-1}
\exp\bigl(-y^{\gamma}\Lambda^{\star}(\phi)\bigr)\,dy.
\]
Proposition 4.1 gives, for $r>t$,
\begin{align*}
\log f_{R,\Phi}(r,\phi)
&=\log p_t+\log f^{(t)}(\phi)+(d-1)\log r
-r^{\gamma}\Lambda^{\star}(\phi)(1+\delta(r,\phi))-\log I^{\delta}_{t}(\phi),\\
\log f^{(t)}_{\mathrm{gen}}(r,\phi)
&=\log p_t+\log f^{(t)}(\phi)+(d-1)\log r
-r^{\gamma}\Lambda^{\star}(\phi)-\log I_{t}(\phi),
\end{align*}
so that
\begin{equation}\label{eq:logdiff}
\log f_{R,\Phi}(r,\phi)-\log f^{(t)}_{\mathrm{gen}}(r,\phi)
=-r^{\gamma}\Lambda^{\star}(\phi)\,\delta(r,\phi)
-\bigl(\log I^{\delta}_{t}(\phi)-\log I_{t}(\phi)\bigr).
\end{equation}
Since $\Lambda^{\star}$ is the limit of continuously converging functions it is itself continuous (\cite{RockafellarWets1998}, Theorem 7.14) and strictly
positive on the compact set $E\cap\mathbb{S}^{d-1}$. Hence, there exist
$0<\lambda_{-}\le\lambda_{+}<\infty$ with
$\lambda_{-}\le\Lambda^{\star}(\phi)\le\lambda_{+}$ for all $\phi$. Set
\[
\bar\delta_t:=\sup_{r>t,\;\phi\in E\cap\mathbb{S}^{d-1}}|\delta(r,\phi)|,
\]
so that $\bar\delta_t\to 0$ as $t\to\infty$. We bound the two terms in
\eqref{eq:logdiff} in turn. For all $r>t$ and $\phi$,
\begin{equation}\label{eq:term1}
\bigl|r^{\gamma}\Lambda^{\star}(\phi)\,\delta(r,\phi)\bigr|
\le \bar\delta_t\,\lambda_{+}\,r^{\gamma}.
\end{equation}
Since $|\delta(y,\phi)|\le\bar\delta_t$ for all $y>t$,
\[
J_t\bigl((1+\bar\delta_t)\Lambda^{\star}(\phi)\bigr)
\;\le\; I^{\delta}_{t}(\phi)\;\le\;
J_t\bigl((1-\bar\delta_t)\Lambda^{\star}(\phi)\bigr),
\qquad
J_t(a):=\int_t^{\infty}y^{d-1}e^{-ay^{\gamma}}\,dy.
\]
The change of variables $x=ay^{\gamma}$ gives
$J_t(a)=\gamma^{-1}a^{-d/\gamma}\,\Gamma(d/\gamma,\,at^{\gamma})$, and the
incomplete-gamma asymptotic $\Gamma(\alpha,x)\sim x^{\alpha-1}e^{-x}$ as
$x\to\infty$, valid uniformly for
$a\in[(1-\bar\delta_t)\lambda_{-},\,(1+\bar\delta_t)\lambda_{+}]$, yields
\[
\log J_t(a)=-at^{\gamma}+(d-\gamma)\log t-\log(\gamma a)+O(t^{-\gamma})
\]
uniformly over this range. Applying this at
$a=(1\pm\bar\delta_t)\Lambda^{\star}(\phi)$ and at $a=\Lambda^{\star}(\phi)$,
and using $I_t(\phi)=J_t(\Lambda^{\star}(\phi))$, there exist $\zeta_t\to 0$
such that
\begin{equation}\label{eq:term2}
\sup_{\phi\in E\cap\mathbb{S}^{d-1}}
\bigl|\log I^{\delta}_{t}(\phi)-\log I_{t}(\phi)\bigr|
\le \zeta_t\,t^{\gamma}.
\end{equation}
Combining \eqref{eq:logdiff}-\eqref{eq:term2}, for all $r>t$ and $\phi$,
\begin{equation}\label{eq:numerator}
\bigl|\log f_{R,\Phi}(r,\phi)-\log f^{(t)}_{\mathrm{gen}}(r,\phi)\bigr|
\le\bigl(\bar\delta_t\lambda_{+}+\zeta_t\bigr)\,(r^{\gamma}+t^{\gamma}).
\end{equation}

\smallskip
\noindent \textbf{Lower bound on the denominator:} Marginalizing over $r$ in Proposition~\ref{prop:Polar-to-Gamma} gives $p_t\,f^{(t)}(\phi)=\int_t^{\infty}f_{R,\Phi}(y,\phi)\,dy$, and by
Assumption~\ref{assume:lt_data},
$f_{R,\Phi}(y,\phi)\le y^{d-1}
\exp\bigl(-(1-\bar\delta_t)y^{\gamma}\Lambda^{\star}(\phi)\bigr)$ for $y>t$.
Hence $p_t f^{(t)}(\phi)\le J_t\bigl((1-\bar\delta_t)\Lambda^{\star}(\phi)\bigr)$,
and the same incomplete-gamma estimate gives, uniformly in $\phi$,
\[
-\log\bigl(p_t\,f^{(t)}(\phi)\bigr)
\;\ge\;(1-\bar\delta_t)\Lambda^{\star}(\phi)\,t^{\gamma}-C\log t
\]
for a constant $C$ depending only on $(d,\gamma,\lambda_{-},\lambda_{+})$.
Together with $\log I_t(\phi)=-\Lambda^{\star}(\phi)t^{\gamma}+O(\log t)$
uniformly in $\phi$, we obtain, for $r>t$,
\begin{align*}
-\log f^{(t)}_{\mathrm{gen}}(r,\phi)
&=-\log\bigl(p_t f^{(t)}(\phi)\bigr)-(d-1)\log r
+\Lambda^{\star}(\phi)\,r^{\gamma}+\log I_{t}(\phi)\\
&\ge \Lambda^{\star}(\phi)\,r^{\gamma}
-\bar\delta_t\,\Lambda^{\star}(\phi)\,t^{\gamma}-C'\log r\\
&\ge (\lambda_{-}-\bar\delta_t\lambda_{+})\,r^{\gamma}-C'\log r,
\end{align*}
where the last step uses $t<r$. Since $\log r/r^{\gamma}$ is decreasing for
$r$ large, for all $t$ sufficiently large we have
$C'\log r\le(\lambda_{-}/4)\,r^{\gamma}$ and
$\bar\delta_t\lambda_{+}\le\lambda_{-}/4$, whence, using $t<r$ once more,
\begin{equation}\label{eq:denominator}
-\log f^{(t)}_{\mathrm{gen}}(r,\phi)
\;\ge\;\frac{\lambda_{-}}{2}\,r^{\gamma}
\;\ge\;\frac{\lambda_{-}}{4}\,(r^{\gamma}+t^{\gamma})\;>\;0
\qquad\text{for all }r>t,\ \phi;
\end{equation}
in particular $\log f^{(t)}_{\mathrm{gen}}(r,\phi)<0$ throughout this range.

\smallskip
Dividing \eqref{eq:numerator} by \eqref{eq:denominator},
\[
\sup_{r>t,\;\phi\in E\cap\mathbb{S}^{d-1}}
\left|\frac{\log f_{R,\Phi}(r,\phi)}{\log f^{(t)}_{\mathrm{gen}}(r,\phi)}-1\right|
=\sup_{r>t,\;\phi}
\frac{\bigl|\log f_{R,\Phi}(r,\phi)-\log f^{(t)}_{\mathrm{gen}}(r,\phi)\bigr|}
{\bigl|\log f^{(t)}_{\mathrm{gen}}(r,\phi)\bigr|}
\le\frac{4\bigl(\bar\delta_t\lambda_{+}+\zeta_t\bigr)}{\lambda_{-}}
\;\longrightarrow\;0. \qquad\square
\]

\noindent \textbf{Proof of Lemma~\ref{lem:lt_correct}:} 
By construction, $\Pphi^\star \sim f^{(t)}$. Conditional on $\Pphi^\star=\pphi$, we generate $Z \sim \Gamma\!\left( d/\gamma,\Lambda^\star(\pphi)\right)$,
and $R=Z^{1/\gamma}$.
A change of variables gives $f_{R\mid \Pphi^\star=\pphi}(r)\propto r^{d-1}\exp\!\bigl(-r^\gamma \Lambda^\star(\pphi)\bigr)$.
Rejecting until $R>t$ 
yields
\[
f_{R_t\mid \Pphi^\star=\pphi}(r)
=
\frac{
r^{d-1}\exp\!\bigl(-r^\gamma \Lambda^\star(\pphi)\bigr)
}{
\int_t^\infty y^{d-1}\exp\!\bigl(-y^\gamma \Lambda^\star(\pphi)\bigr)\,dy
}\mathbf 1(r>t).
\]
Combining this with $\Pphi^\star\sim f^{(t)}$ gives the joint density
\[
f^{(t)}(\pphi)
\frac{
r^{d-1}\exp\!\bigl(-r^\gamma \Lambda^\star(\pphi)\bigr)
}{
\int_t^\infty y^{d-1}\exp\!\bigl(-y^\gamma \Lambda^\star(\pphi)\bigr)\,dy
}\mathbf 1(r>t),
\]
which is precisely the conditional density induced by the tail component of $\fgen^{(t)}$ on $\{r>t\}$. \qed

\noindent \textbf{Proof of Theorem~\ref{thm:est_rare_probs}:}
Write $\Lambda^\star_u:=\inf_{\zz\in\mathcal R_u}\Lambda^\star(\zz)$, and let
\[
    \mathcal L_t(\zz):=\frac{\fgen^{(t)}(\zz)}{f_{\XX}(\zz)},
    \qquad\text{so that}\qquad
    P_{\mathrm{gen}}(\mathcal R_u)=E_P\bigl[\mathcal L_t(\XX);\,\mathcal R_u\bigr].
\]
Note that $\Lambda^\star$ is homogeneous of degree $\gamma$: 
$\Lambda^\star(\zz)\ge\lambda_-\|\zz\|^{\gamma}$, where
$\lambda_-:=\inf_{\pphi\in\mathcal E \cap \mathcal S^{d-1}}\Lambda^\star(\pphi)>0$.

\smallskip
\noindent\textbf{i) Assumption~\ref{assume:ht_data} holds.}
By Theorem~\ref{thm:heavy_vanishing_RE}, the relative error of $\fgen^{(t)}$
vanishes uniformly over the entire state space (the ratio being identically
$1$ on $\{\|\zz\|\le t\}$). Hence for every $\varepsilon>0$ and all
sufficiently large $u$, $1-\varepsilon\le\mathcal L_t(\zz)\le 1+\varepsilon$
for all $\zz$, and therefore
\[
    (1-\varepsilon)P(\mathcal R_u)
    \le
    P_{\mathrm{gen}}(\mathcal R_u)
    \le
    (1+\varepsilon)P(\mathcal R_u).
\]
Since $\varepsilon>0$ is arbitrary,
$P_{\mathrm{gen}}(\mathcal R_u)/P(\mathcal R_u)\to 1$.

\smallskip
\noindent\textbf{ii) Assumption~\ref{assume:lt_data} holds.}
We first note the following consequence of \cite{deo2025achieving}, Theorem 2: for each fixed $a>0$,
\begin{equation}\label{eqn:laplace_sandwich}
\log\int_{\mathcal R_u} f_{\XX}(\zz)^{a}\,d\zz
=
-a\,\Lambda^\star_u\bigl(1+o(1)\bigr),
\qquad u\to\infty .
\end{equation}
Fix $\varepsilon\in(0,1)$. By Assumption~\ref{assume:lt_data}, for all
sufficiently large $u$,
\begin{equation}\label{eqn:pointwise_sandwich}
\exp\bigl(-a(1+\varepsilon)\Lambda^\star(\zz)\bigr)
\le
f_{\XX}(\zz)^{a}
\le
\exp\bigl(-a(1-\varepsilon)\Lambda^\star(\zz)\bigr)
\qquad\text{uniformly over }\zz\in\mathcal R_u .
\end{equation}

\smallskip
\noindent\textbf{Upper bound.} Fix $K>1/(1-\varepsilon)$ and set
$r_u:=\bigl(K\Lambda^\star_u/\lambda_-\bigr)^{1/\gamma}$.
On $\mathcal R_u\cap\{\zz:\|\zz\|\le r_u\}$, the upper bound in
\eqref{eqn:pointwise_sandwich} together with
$\Lambda^\star(\zz)\ge\Lambda^\star_u$ yields
\[
\int_{\mathcal R_u\cap\{\|\zz\|\le r_u\}} f_{\XX}(\zz)^{a}\,d\zz
\;\le\;
c_d\,r_u^{d}\,
\exp\bigl(-a(1-\varepsilon)\Lambda^\star_u\bigr),
\]
where $c_d$ denotes the volume of the unit ball in $\mathbb R^d$; the
logarithm of the right-hand side is
$-a(1-\varepsilon)\Lambda^\star_u+O\bigl(\log\Lambda^\star_u\bigr)$. On the
complementary region, $\Lambda^\star(\zz)\ge\lambda_-\|\zz\|^{\gamma}$, so
passing to polar coordinates and applying the incomplete-gamma asymptotics
from the proof of Theorem~\ref{thm:light_vanishing_LE},
\[
\log\int_{\{\|\zz\|> r_u\}}
\exp\bigl(-a(1-\varepsilon)\lambda_-\|\zz\|^{\gamma}\bigr)\,d\zz
=
-a(1-\varepsilon)\lambda_-\,r_u^{\gamma}\bigl(1+o(1)\bigr)
=
-a(1-\varepsilon)K\,\Lambda^\star_u\bigl(1+o(1)\bigr),
\]
which is of smaller exponential order than the first contribution by the
choice of $K$. Combining the two regions,
\begin{equation}\label{eqn:upper_final}
\log\int_{\mathcal R_u} f_{\XX}(\zz)^{a}\,d\zz
\;\le\;
-a(1-\varepsilon)\,\Lambda^\star_u+o\bigl(\Lambda^\star_u\bigr).
\end{equation}

\smallskip
\noindent \textbf{Lower Bound:} Recall that $u^{-1}L(u\,\cdot)$ converges
continuously to $L^{\star}$ on $\mathcal E$. By
\cite{RockafellarWets1998}, Theorem~7.14, $L^{\star}$ is continuous and the
convergence is uniform on compact subsets of $\mathcal E$ and homogeneous. Since $\Lambda^\star$ is
continuous and $\Lambda^\star(\zz)\ge\lambda_-\|\zz\|^{\gamma}$ for all $\zz$, the closed set
$A:=\{\zz\in\mathcal E:L^{\star}(\zz)\ge1\}$ is bounded away from the
origin, and
$\Lambda^{\star}_{\infty}:=\inf_{\zz\in A}\Lambda^\star(\zz)\in(0,\infty)$
is attained at some $\zz^{\star}\in A$, the infimum reducing to the compact
set $A\cap\{\zz:\|\zz\|\le M\}$ for any
$M^{\gamma}>\Lambda^{\star}_{\infty}/\lambda_{-}$.

Fix $\delta\in(0,1)$ and set $\zz_{\delta}:=(1+\delta)\zz^{\star}$, so that
$L^{\star}(\zz_{\delta})=1+\delta$ by homogeneity. By continuity of
$L^{\star}$ and $\Lambda^\star$, choose $\rho>0$ so that
$D_{\delta}:=B(\zz_{\delta},\rho)\cap\mathcal E$ satisfies
$\inf_{D_{\delta}}L^{\star}\ge1+\delta/2$ and
$\sup_{D_{\delta}}\Lambda^\star\le\bigl((1+\delta)^{\gamma}+\delta\bigr)
\Lambda^{\star}_{\infty}$; since $\mathcal E$ is a full-dimensional closed
cone, $\mathrm{Leb}(D_{\delta})>0$. Uniform convergence of
$u^{-1}L(u\,\cdot)$ on $D_{\delta}$ gives $L(u\yy)>u$ for every
$\yy\in D_{\delta}$ and all large $u$, so that
$B_u:=u\,D_{\delta}\subseteq\mathcal R_u$. By homogeneity of
$\Lambda^\star$,
\begin{equation}\label{eqn:sup_lambda_Bu}
\Lambda^\star_u
\;\le\;
\sup_{\zz\in B_u}\Lambda^\star(\zz)
\;=\;
u^{\gamma}\sup_{\yy\in D_{\delta}}\Lambda^\star(\yy)
\;\le\;
\bigl((1+\delta)^{\gamma}+\delta\bigr)\,\Lambda^{\star}_{\infty}\,u^{\gamma}.
\end{equation}
Conversely, any near-minimizing $\zz_u\in\mathcal R_u$ may be taken to
satisfy $\|\zz_u\|\le Mu$, since
$\Lambda^\star(\zz)\ge\lambda_-M^{\gamma}u^{\gamma}$ otherwise; then
$u^{-1}\zz_u$ lies in a compact set on which $u^{-1}L(u\,\cdot)\to
L^{\star}$ uniformly, so its limit points lie in $A$ and
$\liminf_u u^{-\gamma}\Lambda^{\star}_u\ge\Lambda^{\star}_{\infty}$.
Together with \eqref{eqn:sup_lambda_Bu} and $\delta\downarrow0$, this gives
$u^{-\gamma}\Lambda^{\star}_u\to\Lambda^{\star}_{\infty}$; in particular
$\log\mathrm{Leb}(B_u)=d\log u+O(1)=o(\Lambda^{\star}_u)$ and
$\sup_{B_u}\Lambda^\star\le\bigl((1+\delta)^{\gamma}+\delta\bigr)(1+o(1))
\Lambda^{\star}_u$.

Consequently, by the lower bound in \eqref{eqn:pointwise_sandwich},
\[
\int_{\mathcal R_u} f_{\XX}(\zz)^{a}\,d\zz
\;\ge\;
\mathrm{Leb}(B_u)\,
\exp\Bigl(-a(1+\varepsilon)\sup_{\zz\in B_u}\Lambda^\star(\zz)\Bigr)
\;\ge\;
\exp\Bigl(-a(1+\varepsilon)\bigl((1+\delta)^{\gamma}+\delta\bigr)
\Lambda^{\star}_u\bigl(1+o(1)\bigr)\Bigr).
\]
Combining with \eqref{eqn:upper_final} and letting $\delta\downarrow0$,
$\varepsilon\downarrow0$ proves \eqref{eqn:laplace_sandwich}; in particular,
taking $a=1$,
\begin{equation}\label{eqn:logP_Ru}
\log P(\XX\in\mathcal R_u)=-\Lambda^\star_u\bigl(1+o(1)\bigr).
\end{equation}
If $\zz_u\in\mathcal R_u$ satisfied
$\|\zz_u\|=o(u)$, continuous convergence would give
$u^{-1}L(\zz_u)\to L^\star(\mv 0)=0$; hence $\|\zz\|\ge cu$ on
$\mathcal R_u$ for some $c>0$ and all large $u$, and since $t=o(u)$,
$\mathcal R_u\subset\{\|\zz\|>t\}$.
Fix $\varepsilon\in(0,1)$. 
Theorem~\ref{thm:light_vanishing_LE} gives for all sufficiently large $u$,
\[
\sup_{\zz\in\mathcal R_u}
\left|\frac{\log\fgen^{(t)}(\zz)}{\log f_{\XX}(\zz)}-1\right|
\le\varepsilon,
\]
with $\log f_{\XX}(\zz)<0$ on $\mathcal R_u$. Writing
$\log\mathcal L_t(\zz)
=\bigl(\log\fgen^{(t)}(\zz)/\log f_{\XX}(\zz)-1\bigr)\log f_{\XX}(\zz)$,
it follows that on $\mathcal R_u$,
\[
f_{\XX}(\zz)^{\varepsilon}
\;\le\;
\mathcal L_t(\zz)
\;\le\;
f_{\XX}(\zz)^{-\varepsilon}.
\]
Multiplying by $f_{\XX}$ and integrating over $\mathcal R_u$,
\[
\int_{\mathcal R_u} f_{\XX}(\zz)^{1+\varepsilon}\,d\zz
\;\le\;
P_{\mathrm{gen}}(\mathcal R_u)
\;\le\;
\int_{\mathcal R_u} f_{\XX}(\zz)^{1-\varepsilon}\,d\zz,
\]
and \eqref{eqn:laplace_sandwich} with $a=1\pm\varepsilon$, together with
\eqref{eqn:logP_Ru}, yields
\[
(1+\varepsilon)\bigl(1+o(1)\bigr)\log P(\XX\in\mathcal R_u)
\;\le\;
\log P_{\mathrm{gen}}(\mathcal R_u)
\;\le\;
(1-\varepsilon)\bigl(1+o(1)\bigr)\log P(\XX\in\mathcal R_u).
\]
Since $\log P(\XX\in\mathcal R_u)<0$, dividing by it reverses the
inequalities, giving
\[
1-\varepsilon+o(1)
\;\le\;
\frac{\log P_{\mathrm{gen}}(\mathcal R_u)}{\log P(\XX\in\mathcal R_u)}
\;\le\;
1+\varepsilon+o(1).
\]
Letting $\varepsilon\downarrow0$ completes the proof.\qed

\section{Implementation Details}

\noindent\textbf{GAN Implementation Details.}
We employ a fully connected GAN to learn the angular exceedance distribution. Let
\[
\Theta_t = \left\{ \theta_i = \frac{X_i}{\|X_i\|} : \|X_i\| \geq t \right\} \subset \mathbb{S}^{d-1}
\]
denote the empirical angular sample exceeding the radial threshold~$t$. The generator $G_\eta : \mathbb{R}^{d_z} \to \mathbb{R}^d$ maps latent variables $z \sim \nu_z$ to ambient angular vectors, where the latent dimension is $d_z = d - 1$ and $\nu_z$ is the uniform distribution on the $d_z$-dimensional unit ball. The discriminator $D_\omega : \mathbb{R}^d \to [0, 1]$ distinguishes empirical from generated angular samples via the standard minimax objective:
\[
\min_{\eta} \max_{\omega} \left\{ \mathbb{E}_{\theta \sim \widehat{P}_t} [\log D_\omega(\theta)] + \mathbb{E}_{z \sim \nu_z} [\log(1 - D_\omega(G_\eta(z)))] \right\},
\]
where $\widehat{P}_t$ is the empirical distribution of $\Theta_t$.

Both networks are implemented as multilayer perceptrons. The generator uses hidden layer widths of $128, 64, 32, 16, 8$, and $4$, with each hidden layer consisting of a dense affine transformation followed by Leaky ReLU activation (negative slope $0.2$) and batch normalization. The final output layer is a dense mapping from the last hidden width back to the full ambient dimension~$d$, with a \texttt{tanh} activation; the decreasing hidden widths form a bottleneck that forces the network to learn a compressed representation of the angular law before projecting back to $\mathbb{R}^d$. The discriminator mirrors the generator's hidden widths, using dense layers followed by Leaky ReLU activations (negative slope $0.2$) without batch normalization, and terminates in a single sigmoid output.

To ensure the generated samples reside on the unit sphere, the raw generator output $\widetilde{\theta} = G_\eta(z)$ is $\ell^2$-normalized via the projection
\[
\Pi(\widetilde{\theta}) = \frac{\widetilde{\theta}}{\|\widetilde{\theta}\|},
\]
so that the effective angular generator used by SSGEN is the composition $\Pi \circ G_\eta$.

Both networks are optimized using Adam ($\mathrm{lr}=10^{-4}$, $\beta_1=0.5$) with the binary cross-entropy loss. Each training iteration consists of one discriminator update on a batch of real angular samples followed by a batch of generated samples, and one generator update through the frozen discriminator. To monitor convergence, we evaluate a random-projection diagnostic every $1{,}000$ epochs: both the empirical angular sample and a batch of generated samples (each $\ell^2$-normalized) are projected onto $256$ random unit vectors in $\mathbb{R}^d$, and the one-dimensional Wasserstein-1 and Kolmogorov--Smirnov distances are computed on each projection. The mean projected Wasserstein distance serves as the primary fidelity metric for model selection.

Training proceeds for a maximum of $30{,}000$ epochs. Figure~\ref{fig:gan_training} shows representative training curves for the $d=10$ mixture experiment. The projected Wasserstein and KS distances reach their minimum during the early phase of training (approximately epochs $4{,}000$--$7{,}000$). Beyond this point, the discriminator loss decreases to near zero while the generator loss rises, indicating that the discriminator begins to dominate - this is
a well-known instability in GAN training. The projection-based diagnostics deteriorate accordingly. We therefore select the generator checkpoint corresponding to the lowest mean projected Wasserstein distance rather than using the final iterate. This checkpoint-based model selection strategy ensures that SS-GEN uses the best angular generator encountered during training.

\begin{figure}[htbp]
    \centering
    \includegraphics[width=0.95\linewidth]{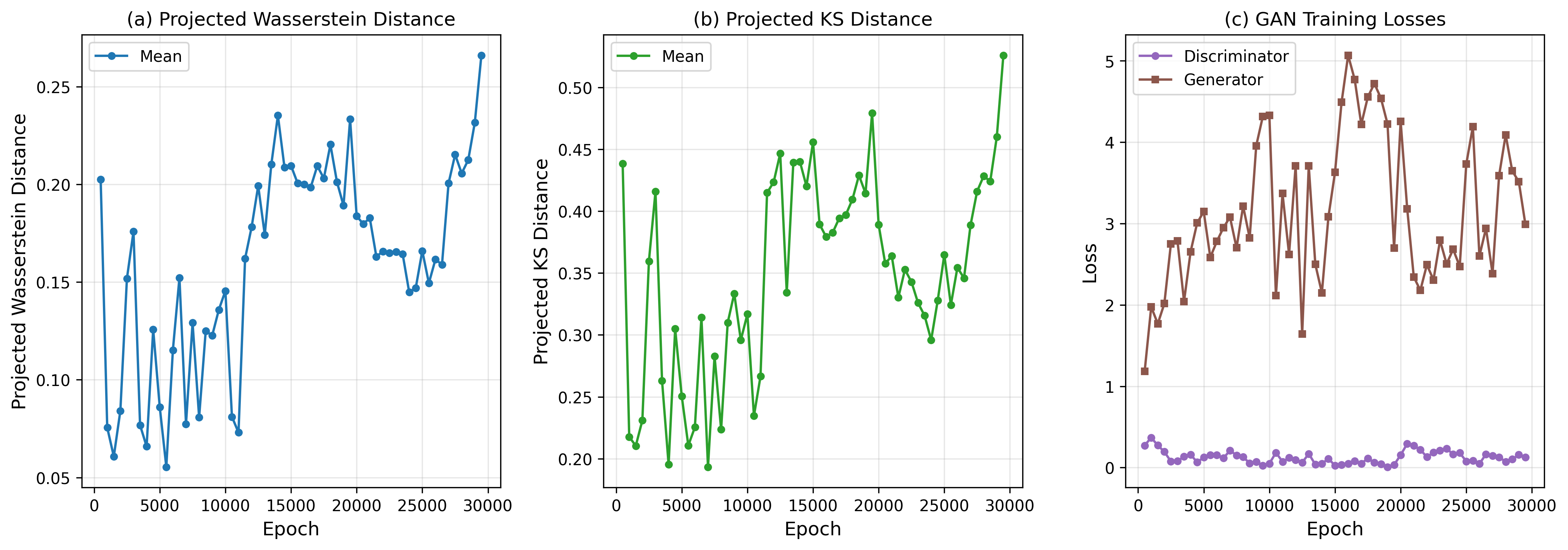}
    \caption{GAN training diagnostics for the $d=10$ mixture experiment. (a)~Mean projected Wasserstein distance across $256$ random projections. (b)~Corresponding projected KS distances. (c)~Generator and discriminator losses.}
    \label{fig:gan_training}
\end{figure} 
\noindent\textbf{VaR and CVaR Estimation Algorithm.} 
Algorithm~\ref{algo:var_cvar_ssgen_ht} details the procedure for computing
portfolio-level VaR and CVaR from SS-GEN samples in Experiment~1
(Section~\ref{sec:numericals}). The construction is stated for the
heavy-tailed case (Assumption~\ref{assume:ht_data}), consistent with the
scope of the experiments; the Weibull-type analogue replaces the Pareto
radial draw in Step~I(iii) with the truncated-Gamma draw of
Algorithm~\ref{algo:generate_extremes_lt}.

\begin{algorithm}[h]
  \caption{Computation of VaR and CVaR using \SSGEN}
  \label{algo:var_cvar_ssgen_ht}
  \KwIn{Data samples $\XX_1,\ldots,\XX_n$, portfolio weights $\xx$,
  threshold $t$, confidence level $\alpha$, number of synthetic samples $M$,
  loss function $L$, trained angular model from
  Algorithm~\ref{algo:generate_extremes}, tail-index estimate $\widehat s$}



\noindent\textbf{I) Generation Step}
\begin{enumerate}
    \item[(i)] Define the non-tail index set, synthetic index set, and empirical tail probability as
    \[
        \mathcal{Y}^c = \{1,\ldots,n\}\setminus \mathcal{Y},
        \qquad
        \mathcal{M} = \{1,\ldots,M\},
        \qquad
        \hat p = \frac{|\mathcal{Y}|}{n}.
    \]

    \item[(ii)] For each $j \in \mathcal{M}$, sample $\Pphi_j^*$ from the trained DGM.

    \item[(iii)] For each $j \in \mathcal{M}$, sample $R_{t,j}$ from the Pareto exceedance density
    \[
        s t^s r^{-(s+1)}\mathbf{1}(r>t).
    \]

    \item[(iv)] Construct the synthetic tail samples
    \[
        H_{t,j} = R_{t,j}\cdot \Pphi_j^*,
        \qquad j \in \mathcal{M}.
    \]
\end{enumerate}

\noindent\textbf{II) Augmented Empirical Loss Distribution}
\begin{enumerate}
    \item[(i)] Define the augmented empirical loss distribution as
    \begin{equation}
    \label{eqn:augmented_cdf_loss_ssgen_ht}
    \widehat F_L^{\SSGEN}(\ell)
    =
    \frac{1}{n}
    \sum_{i \in \mathcal{Y}^c}
    \mathbf{1}\{L(\xx,\XX_i)\leq \ell\}
    +
    \hat p \frac{1}{M}
    \sum_{j \in \mathcal{M}}
    \mathbf{1}\{L(\xx,H_{t,j})\leq \ell\}.
    \end{equation}
\end{enumerate}

\noindent\textbf{III) Estimation Step}
\begin{enumerate}
    \item[(i)] Estimate the Value-at-Risk at level $\alpha$ by
    \[
        \widehat{\mathrm{VaR}}_{\alpha,\SSGEN}(\xx)
        =
        \inf
        \left\{
        \ell \in \mathbb{R} :
        \widehat F_L^{\SSGEN}(\ell)
        \geq \alpha
        \right\}.
    \]

    \item[(ii)] Estimate the Conditional Value-at-Risk at level $\alpha$ by
    \begin{equation}
    \label{eqn:cvar_augmented_empirical_loss_ssgen_ht}
    \begin{split}
    \widehat{\mathrm{CVaR}}_{\alpha,\SSGEN}(\xx)
    &=
    \frac{1}{1-\alpha}
    \Bigg[
    \sum_{i \in \mathcal{Y}^c}
    \frac{1}{n}
    L(\xx,\XX_i)
    \mathbf{1}
    \left\{
    L(\xx,\XX_i)
    \geq
    \widehat{\mathrm{VaR}}_{\alpha,\SSGEN}(\xx)
    \right\}
    \\
    &\qquad
    +
    \sum_{j \in \mathcal{M}}
    \frac{\hat p}{M}
    L(\xx,H_{t,j})
    \mathbf{1}
    \left\{
    L(\xx,H_{t,j})
    \geq
    \widehat{\mathrm{VaR}}_{\alpha,\SSGEN}(\xx)
    \right\}
    \Bigg].
    \end{split}
    \end{equation}
\end{enumerate}

\noindent\textbf{Output: }
$\widehat{\mathrm{VaR}}_{\alpha,\SSGEN}(\xx)$ and
$\widehat{\mathrm{CVaR}}_{\alpha,\SSGEN}(\xx)$.

\end{algorithm}

\end{document}